\documentclass[preprint]{elsarticle}
\usepackage[utf8]{inputenc}
\usepackage{amsmath,amssymb,amsthm,mathtools,nicefrac}
\usepackage{tikz}

\newtheorem{observation}{Observation}
\newtheorem{example}{Example}
\newtheorem{theorem}{Theorem}
\newtheorem{lemma}{Lemma}
\newtheorem{algorithm}{Algorithm}

\newcommand{\sources}{S^+}
\newcommand{\sinks}{S^-}

\title{A Note on the Quickest Minimum Cost Transshipment Problem\tnoteref{t1}}

\tikzstyle{node}=[circle, inner sep = 0pt, minimum size = 0.9em, fill, fill = black, text = white]
\tikzstyle{arc}=[ultra thick, -stealth]

\tnotetext[t1]{Supported by the Deutsche Forschungsgemeinschaft (DFG, German Research Foundation) under Germany's Excellence Strategy --- The Berlin Mathematics Research Center MATH+ (EXC-2046/1, project ID: 390685689).}

\author[1]{Martin Skutella}
\ead{martin.skutella@tu-berlin.de}

\address[1]{TU Berlin, Institut für Mathematik, Str.~des 17.~Juni 136, 10623 Berlin, Germany, martin.skutella@tu-berlin.de}

\begin{document}

\begin{abstract}
Klinz and Woeginger (1995) prove that the minimum cost quickest flow problem is NP-hard. On the
other hand, the quickest minimum cost flow problem can be solved efficiently via a
straightforward reduction to the quickest flow problem without costs. More generally, we show how
the quickest minimum cost transshipment problem can be reduced to the efficiently solvable
quickest transshipment problem, thus adding another mosaic tile to the rich complexity landscape
of flows over time.
\end{abstract}

\begin{keyword}
flow over time \sep transshipment \sep transportation problem \sep complexity 
\end{keyword}

\maketitle
\section{Introduction}

Network flows over time generalize the classical concept of static network flows by incorporating
the temporal development of flow progressing through a network with transit times on the arcs.
They have numerous applications in various areas such as traffic and transport, production,
finance, evacuation, and communication; see, e.g., the classical surveys by \citet{Aronson89} and
by \citet{PowellJO95}.

Historically, flows over time have been introduced by \citet{FordFulkerson58} (see also their
classical textbook~\cite{FordFulkerson62}). For single-source single-sink networks with
capacities and transit times on the arcs, they show how to efficiently compute a maximum flow
over time with given time horizon. If, instead of fixing the time horizon, we fix the flow value
and ask for a flow over time with minimum time horizon, we arrive at the \emph{quickest flow
problem}. For this problem, \citet{Saho2017} present the currently fastest known algorithm with
strongly polynomial running time~$O(m^2n\log^2 n)$, where~$m$ denotes the number of arcs and~$n$
the number of nodes.

Somewhat surprisingly, the \emph{quickest transshipment problem} in networks containing several
sources with given supplies and several sinks with given demands seems to require much more
advanced algorithmic techniques than the single-source single-sink case. The first algorithm due
to \citet{HoppeTardos00}, as well as all efficient algorithms developed subsequently, rely on
parametric submodular function minimization; see~\citet{SchloterST22} for the latest and
currently fastest such algorithm.

Maybe even more surprisingly, \citet{KlinzW95,KlinzW04} reveal a significant complexity gap
between static flows and flows over time. They show that it is NP-hard to compute minimum cost
$s$-$t$-flows over time in networks with capacities, transit times, and cost coefficients on the
arcs. In particular, they prove that the \emph{minimum cost quickest flow problem}, that is,
finding among all flows over time with minimum time horizon one that has minimum cost, is NP-hard.

Reversing the order of the two objectives time and cost yields the \emph{quickest minimum cost flow
problem} 
which, to the best of our knowledge, is considered for the first time in this paper, and constitutes
a special case of the problem in the main focus of this paper.
Here, the primary objective is to minimize the total cost of the flow over time, and the
secondary objective is to minimize its time horizon. 
It is easy to observe that,
in this case, flow may be sent along cheapest
source-sink-paths only. The quickest minimum cost flow problem thus polynomially reduces to the
efficiently solvable quickest flow problem on the subnetwork formed by all arcs that are contained
in some cheapest path. This rather straightforward observation is discussed in some more detail in
Section~\ref{sec:quickest-min-cost-flow}. 

In Section~\ref{sec:quickest-min-cost-transshipment} we present a generalization of this polynomial
reduction to the \emph{quickest minimum cost transshipment problem}. We prove that also a quickest
minimum cost transshipment can be obtained by computing a quickest transshipment in a suitably
chosen subnetwork. Determining this subnetwork, however, turns out to be nontrivial as it requires
an optimal dual solution to a static transportation problem on a bipartite graph.
%
\subsection*{Preliminaries and notation}
We consider a directed graph~$D=(V,A)$ with node set~$V$ and arc set~$A$. We use~$n\coloneqq|V|$ to
denote the number of nodes, and~$m\coloneqq|A|$ to denote the number of arcs. Every arc~$a\in A$
has a positive \emph{capacity~$u_a>0$}, a non-negative \emph{transit time~$\tau_a\geq0$}, and a
\emph{cost coefficient~$c_a$}. For a path~$P$ in~$D$, we denote its total transit time
by~$\tau(P)\coloneqq\sum_{a\in P}\tau_a$, and its total cost by~$c(P)\coloneqq\sum_{a\in P}c_a$. We
assume throughout that arc costs are \emph{conservative}, that is, there is no negative cost cycle in~$D$;
otherwise, there might not exist a minimum cost flow over time since cost can be decreased
arbitrarily by repeatedly sending flow through a negative cost cycle.

Among the nodes in~$D$, there is a subset of \emph{sources~$\sources\subset V$} and a subset of
\emph{sinks~$\sinks\subseteq V\setminus\sources$} with given \emph{supplies~$b_s>0$},
for~$s\in\sources$, and \emph{demands~$-b_t>0$}, for~$t\in\sinks$. As usual, we require that the
total supply equals the total demand, that is,~$\sum_{s\in\sources}b_s+\sum_{t\in \sinks}b_t=0$.

For the purpose of this note, only a rudimentary understanding of flows over time is necessary. A
thorough introduction to the area can be found in the survey~\cite{Skutella-Korte09}. A flow over
time~$f$ with time horizon~$\theta$ specifies for each arc~$a\in A$ a Lebesgue-measurable
function~$f_a:[0,\theta-\tau_a)\to\mathbb{R}_{\geq0}$, where~$f_a(\theta')$ is the flow rate
entering arc~$a$ at its tail at time~$\theta'$. Thus, the total amount of flow entering and
traversing arc~$a$ is~$\int_0^{\theta-\tau_a}f_a(\theta')\,\text{d}\theta'$, and the cost incurred
on arc~$a$ is~$c_a$ times this amount.

The capacity~$u_a$ of arc~$a$ provides an upper bound on its inflow rate at all points in
time:~$f_a(\theta')\leq u_a$, for all~$\theta'\in[0,\theta-\tau_a)$. The transit time~$\tau_a$ is
the time it takes to traverse arc~$a$. In particular, the outflow rate at the head of arc~$a$ at
time~$\theta'\in[\tau_a,\theta)$ is~$f_a(\theta'-\tau_a)$. \emph{Flow conservation} states that, with
the exception of sources in~$\sources$, no flow deficit may occur at nodes at any point in time.
Moreover, with the exception of sinks in~$\sinks$, no flow surplus must remain at nodes at
time~$\theta$. Depending on the specifics of the considered flow model, it might or might not be
allowed to temporarily store flow at intermediate nodes. But \citet{FleischerS07} prove that the
possibility of temporally storing flow at intermediate nodes does not lead to cheaper
transshipments over time.

We conclude these introductory remarks on flows over time with a small, illustrative example:
Sending~$b>0$ units of flow from node~$s$ to node~$t$ along a given~$s$-$t$-path~$P$ with
bottleneck capacity~$u\coloneqq\min_{a\in P}u_a$ can, for example, be achieved within time
horizon~$b/u+\tau(P)$ by sending flow at constant rate~$u$ into~$P$ during the time
interval~$[0,b/u)$ such that the last flow particle reaches~$t$ before time~~$b/u+\tau(P)$.

\section{Quickest minimum cost $s$-$t$-flows}
\label{sec:quickest-min-cost-flow}
As a warm-up and easy exercise, we consider the case of a single source node~$s$ and a single sink
node~$t$, with supply/demand~$b_s=-b_t$. We want to find a quickest minimum cost $s$-$t$-flow
in~$D$, that is, we are looking for an $s$-$t$-flow over time that sends~$b_s$ units of flow
from~$s$ to~$t$ in a cheapest possible way and, among all such minimum cost solutions, has minimum
time horizon~$\theta$.

If we do \emph{not} restrict the time horizon~$\theta$, a cheapest way of sending~$b_s=-b_t$
units of flow from source~$s$ to sink~$t$ is to send the entire flow along a cheapest
$s$-$t$-path~$P$, resulting in overall cost~$b_s\cdot c(P)$. As long as time is not an issue,
capacity values and transit times of arcs on path~$P$ do not play a role. They only come into
play when we ask for a minimum cost~$s$-$t$-flow over time with bounded time horizon, or even a
quickest minimum cost~$s$-$t$-flow. In that case, it might be beneficial to use more than only one
cheapest path.

Notice that an $s$-$t$-flow over time has minimum cost if and only if all flow (but a null set)
is sent along cheapest $s$-$t$-paths. That is, all flow is sent through the cheapest-paths
network~$D'=(V,A')$, which consists of node set~$V$ and arc subset
\[
A'\coloneqq\{a\in A\mid\text{$a$ lies on some cheapest~$s$-$t$-path in~$D$}\}. 
\]
This immediately yields the following observation.
\begin{observation}
A quickest minimum cost $s$-$t$-flow can be obtained by finding a quickest $s$-$t$-flow in the
cheapest-paths network~$D'=(V,A')$.
\end{observation}
The currently best-known running time for computing a quickest $s$-$t$-flow is $O(m^2n\log^2n)$
due to~\citet{Saho2017}. It dominates the time required to compute the cheapest-paths
network~$D'=(V,A')$
defined above, 
which can easily be obtained 
as follows: For all~$v\in V$, compute the cost~$\alpha_v$ of a cheapest~$s$-$v$-path as well as the cost~$\beta_v$ of a cheapest~$v$-$t$-path in~$D$. Then, the cheapest-paths
network's set of arcs is~$A'=\{(v,w)\in A\mid \alpha_v+c_{(v,w)}+\beta_w=\alpha_t\}$.
%

\section{Quickest minimum cost transshipments}
\label{sec:quickest-min-cost-transshipment}
When we turn to the quickest minimum cost transshipment problem in networks with several sources
and sinks, it is still true that an optimum solution must only use arcs that lie on a cheapest
$s$-$t$-path for some source~$s\in\sources$ and sink~$t\in\sinks$.

In contrast to the special case of $s$-$t$-flows, however, this necessary condition is no longer
sufficient. Instead, as we argue below, a quickest minimum cost transshipment can be
found by computing a quickest transshipment in a more carefully chosen subnetwork. Before going
into details, we provide an illustrating example.
\begin{example}
\label{exampl}
We consider the network~$D=(V,A)$ depicted in Figure~\ref{fig:counterexample}, with two source
nodes~$s_1,s_2$ of unit supply~$b_{s_1}=b_{s_2}=1$ and two sink nodes~$t_1,t_2$ of unit
demand~$-b_{t_1}=-b_{t_2}=1$.
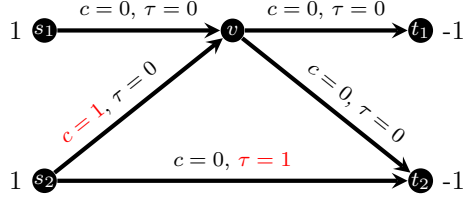
\begin{figure}[t]
\centering
\begin{tikzpicture}
\node [node,label=left:1] (s1) at (-2.5,0) {\footnotesize$s_1$};
\node [node,label=left:1] (s2) at (-2.5,-2) {\footnotesize$s_2$};
\node [node] (v) at (0,0) {\footnotesize$v$};
\node [node,label=right:-1] (t1) at (2.5,0) {\footnotesize$t_1$};
\node [node,label=right:-1] (t2) at (2.5,-2) {\footnotesize$t_2$};
\draw [arc] (s1) to node [above] {\footnotesize$c=0$, $\tau=0$} (v);
\draw [arc] (s2) to node [above, sloped, pos=0.4] {\footnotesize$\color{red}c=1$, $\tau=0$} (v);
\draw [arc] (s2) to node [above] {\footnotesize$c=0$, $\color{red}\tau=1$} (t2);
\draw [arc] (v) to node [above] {\footnotesize$c=0$, $\tau=0$} (t1);
\draw [arc] (v) to node [above, sloped, pos=0.6] {\footnotesize$c=0$, $\tau=0$} (t2);
\end{tikzpicture}
\caption{Quickest minimum cost transshipment instance with unit arc capacities, two sources with
unit supplies (left), and two sinks with unit demands (right)}
\label{fig:counterexample}
\end{figure}
Notice that the network is constructed such that every arc~$a\in A$ lies on a
cheapest~$s_i$-$t_j$-path for some pair~$i,j\in\{1,2\}$.

A quickest transshipment in~$D$ sends flow at rate~$1$ between times~$0$ and~$1$ through each zero
transit time arc. It has time horizon~$\theta=1$ and total cost~$1$, caused by sending one unit of
flow through the only `expensive' arc~$s_2v$ with cost coefficient~$c_{s_2v}=1$. A quickest
\emph{minimum cost} transshipment, however, may not use the costly arc~$s_2v$. Instead, it needs to
send one unit of flow from~$s_2$ to~$t_2$ via the direct but more time-consuming arc~$s_2t_2$ with
transit time~$\tau_{s_2t_2}=1$, and another unit from~$s_1$ to~$t_1$ via node~$v$. It thus requires
time horizon~$\theta=2$, total cost~$0$, and is a quickest transshipment in a
subnetwork~$D'=(V,A')$ where~$A'$ is obtained from~$A$ by deleting arc~$s_2v$ or arc~$vt_2$ (or
both).
\end{example}
In general, the correct choice of subnetwork~$D'=(V,A')$ such that any quickest transshipment
in~$D'$ is a quickest minimum cost transshipment in the original network~$D$ not only depends on
the arc costs, but also on the given supplies and demands at the sources and sinks, respectively.
This can be illustrated again using the network depicted in Figure~\ref{fig:counterexample}.

If we increase the supply of source~$s_2$ to~$b_{s_2}=\nicefrac32$ and the demand of sink~$t_1$
to~$-b_{t_1}=\nicefrac32$, any 
%
%
transshipment over time must send half a unit of flow
through the costly arc~$s_2v$. In this case, a quickest minimum cost transshipment still has
time horizon~$\theta=2$, total cost~$\nicefrac12$, and is a quickest transshipment in the
subnetwork obtained by only deleting arc~$vt_2$, that is,~\mbox{$A'=A\setminus\{vt_2\}$}.

On the other hand, if instead we increase the supply of source~$s_1$ to~$b_{s_1}=\nicefrac32$ and
the demand of sink~$t_2$ to~$-b_{t_2}=\nicefrac32$ in Figure~\ref{fig:counterexample}, a minimum
cost transshipment over time may no longer use the costly arc~$s_2v$, but must send half a unit of
flow through arc~$vt_2$. Thus, a quickest minimum cost transshipment has time horizon~$\theta=2$
again, total cost~$0$, and is a quickest transshipment in the subnetwork obtained by deleting
arc~$s_2v$ only, i.e.,~$A'=A\setminus\{s_2v\}$.

For an arbitrary instance of the quickest minimum cost transshipment problem, we would like to
determine which arcs of the given network~$D=(V,A)$ may be used by a minimum cost transshipment
over time. As a first step, we ask whether a particular source node~$s\in\sources$ may serve a
particular sink node~$t\in\sinks$ in a minimum cost transshipment over time.

To this end, we construct a static transportation problem on a bipartite
network~$\bar{D}=(\sources\cup\sinks,\bar{A})$ with~$\bar{A}\coloneqq\bigl\{st\in\sources\times
\sinks\mid\text{$\exists$ $s$-$t$-path in~$D$}\bigr\}$. An arc~$st\in\bar{A}$ thus represents the
option to send flow from source~$s$ to sink~$t$ in a transshipment over time. Since a
transshipment over time in~$D$ with arbitrarily large time horizon may send an arbitrary amount
of flow along an~$s$-$t$-path, all arcs in~$\bar{A}$ get infinite capacity. Moreover, the
cost~$\bar{c}_{st}$ of arc $st\in\bar{A}$ is set to the cost of a cheapest $s$-$t$-path in~$D$,
since, as already discussed above, a minimum cost transshipment over time may send flow along
cheapest paths only. Finally, supplies and demands at the nodes of the bipartite
network~$\bar{D}$ match those given in the original network. In Figure~\ref{fig:Dbar} we
illustrate the transportation problem corresponding to Example~\ref{exampl}, where the supply of source~$s_2$ and the demand of sink~$t_1$ have been increased to~$b_{s_2}=-b_{t_1}=\nicefrac32$.
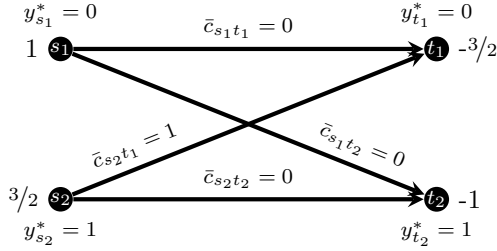
\begin{figure}[t]
\centering
\begin{tikzpicture}
\node [node,label=left:1,label=above:{\footnotesize$y^*_{s_1}=0$}] (s1) at (-2.5,0) {\footnotesize$s_1$};
\node [node,label=left:\nicefrac32,label=below:{\footnotesize$y^*_{s_2}=1$}] (s2) at (-2.5,-2) {\footnotesize$s_2$};
\node [node,label=right:-\nicefrac32,label=above:{\footnotesize$y^*_{t_1}=0$}] (t1) at (2.5,0) {\footnotesize$t_1$};
\node [node,label=right:-1,label=below:{\footnotesize$y^*_{t_2}=1$}] (t2) at (2.5,-2) {\footnotesize$t_2$};
\draw [arc] (s1) to node [above] {\footnotesize$\bar{c}_{s_1t_1}=0$} (t1);
\draw [arc] (s1) to node [above, sloped, pos=0.8] {\footnotesize$\bar{c}_{s_1t_2}=0$} (t2);
\draw [arc] (s2) to node [above, sloped, pos=0.2] {\footnotesize$\bar{c}_{s_2t_1}=1$} (t1);
\draw [arc] (s2) to node [above] {\footnotesize$\bar{c}_{s_2t_2}=0$} (t2);
\end{tikzpicture}
\caption{Transportation problem corresponding to the quickest minimum cost transshipment instance
depicted in Fig.~\ref{fig:counterexample} (with increased
supply/demand~$b_{s_2}=-b_{t_1}=\nicefrac32$), together with an optimum dual solution~$y^*$}
\label{fig:Dbar}
\end{figure}
\begin{observation}
\label{obs:transportation}
The static transportation problem on~$\bar{D}$ has a feasible solution if and only if there
exists a transshipment over time in~$D$. A cheapest solution to the static transportation problem
on~$\bar{D}$ has the same cost as a minimum cost transshipment over time in~$D$. Moreover, there
is a minimum cost transshipment over time in~$D$ that sends a positive amount of flow from
source~$s\in\sources$ to sink~$t\in\sinks$ if and only if~$st\in\bar{A}$ and there is an optimal
solution to the static transportation problem with positive flow value on arc~$st\in\bar{A}$.
\end{observation}
\begin{proof}
Given a transshipment over time in~$D$, a solution to the transportation problem on~$\bar{D}$ can
be obtained as follows: assign to each arc~$st\in\bar{A}$ the amount of flow that is being sent
from~$s$ to~$t$ in the given transshipment over time in~$D$. By definition of~$\bar{c}_{st}$, the
cost for sending this flow through arc~$st$ is a lower bound on the respective cost incurred by
the transshipment over time in~$D$.

Vice versa, given a solution to the transportation problem on~$\bar{D}$, one can construct an
equally expensive transshipment over time in~$D$ as follows: for each arc~$st\in\bar{A}$, send
the amount of flow carried by~$st$ along a cheapest~$s$-$t$-path in~$D$. Since the time
horizon~$\theta$ may be chosen arbitrarily, the entire transport of flow through~$D$ can be
scheduled such that no arc capacities are violated.
\end{proof}
In order to exploit this observation, we consider the natural linear programming formulation of the static transportation problem on~$\bar{D}$. For every~$st\in\bar{A}$, the variable~$x_{st}$ gives the amount of flow sent through arc~$st$ in~$\bar{D}$: 
\begin{align*}
\min~
&\sum_{st\in\bar{A}}\bar{c}_{st}\cdot x_{st}\\
\text{s.t.}\quad
&\sum_{t:st\in\bar{A}}x_{st}~=~b_s && \text{for all~$s\in\sources$,}\\
&\sum_{s:st\in\bar{A}}-x_{st}~=~b_t && \text{for all~$t\in\sinks$,}\\
&x_{st}~\geq~0 && \text{for all~$st\in\bar{A}$.}
\end{align*}
The corresponding dual linear program is:
\begin{align}
\max~
&\sum_{s\in\sources}b_s\cdot y_s+\sum_{t\in\sinks}b_t\cdot y_t\notag\\
\text{s.t.}\quad
& y_s-y_t~\leq~\bar{c}_{st} && \text{for all~$st\in\bar{A}$.}\label{eq:dual}
\end{align}
In view of Observation~\ref{obs:transportation}, we assume that there exists a transshipment over
time in~$D$ and thus a feasible solution to our static transportation problem. Since we assume
that cost coefficients in~$D$ are conservative, Observation~\ref{obs:transportation} also implies
the existence of minimum cost solutions to both problems. In particular, the dual linear program
above has an optimum solution which we denote by~$y^*$. For~$s\in\sources$ and~$t\in\sinks$, we
say that the source-sink pair~$s,t$ is \emph{active}, if~$st\in\bar{A}$ and the corresponding
dual constraint~\eqref{eq:dual} is tight, i.e., $y^*_s-y^*_t=\bar{c}_{st}$. Moreover, an
$s$-$t$-path~$P$ in~$D$ is called \emph{admissible} if $s,t$ is an active pair and~$P$ is a
cheapest $s$-$t$-path, i.e.,~$c(P)=\bar{c}_{st}$.
\begin{observation}
\label{obs:active}
A transshipment over time in~$D$ has minimum cost if and only if flow is being sent from~$s\in\sources$ to~$t\in\sinks$ along admissible $s$-$t$-paths only.
\end{observation}
\begin{proof}
The observation is a direct consequence of Observation~\ref{obs:transportation} and the
complementary slackness theorem of linear programming.
\end{proof}
Next we show that there is a sub-network~$D'=(V,A')$ of~$D$ that contains all
admissible~$s$-$t$-paths, and such that any~$s$-$t$-path in~$D'$ is admissible, for~$s\in\sources$
and~$t\in\sinks$. To this end, starting from~$D$, we first construct an extended
network~$\tilde{D}$ by adding a super-source~$\tilde{s}$ and a super-sink~$\tilde{t}$ to node
set~$V$. The super-source~$\tilde{s}$ is connect to all sources~$s\in\sources$ by
arcs~$\tilde{s}s$; and all sinks~$t\in\sinks$ are connected to the super-sink~$\tilde{t}$ by
arcs~$t\tilde{t}$.
%
%
The new arcs' cost coefficients are set to~$c_{\tilde{s}s}\coloneqq-y^*_s$, for~$s\in\sources$,
and~$c_{t\tilde{t}}\coloneqq y^*_t$, for~$t\in\sinks$. Notice that, by
construction,~$\tilde{D}$ does not contain negative cost cycles since every cycle in~$\tilde{D}$
is also contained in~$D$.
\begin{lemma}
\label{lem:active}
Consider the subnetwork~$D'=(V,A')$ of~$D$ with
\[
A'\coloneqq\bigl\{a\in A\mid\text{$a$ lies on cheapest~$\tilde{s}$-$\tilde{t}$-path in~$\tilde{D}$}\bigr\}.
\]
Then, for~$s\in\sources$ and~$t\in\sinks$, 
an $s$-$t$-path in~$D$ is admissible if and only if it is contained in subnetwork~$D'$,
%
that is, if and only if it uses edges in~$A'$ only.
\end{lemma}
\begin{proof}
Any $s$-$t$-path~$P$ in~$D$ naturally corresponds to an~$\tilde{s}$-$\tilde{t}$-path~$\tilde{P}$
in the extended network~$\tilde{D}$, where~$\tilde{P}$ is obtained from~$P$ by adding
arcs~$\tilde{s}s$ and~$t\tilde{t}$ in the beginning and end, respectively.

We start off by showing that~$P$ is admissible if and only if~$c(\tilde{P})=0$. Notice that the
cost of~$\tilde{P}$ is equal to $c_{\tilde{s}s}+c(P)+c_{t\tilde{t}} = -y^*_s+c(P)+y^*_t$.
Thus,~$c(\tilde{P})=0$ if and only if~$c(P)=y^*_s-y^*_t$. The latter condition is equivalent to
the definition of~$P$ being admissible.

It remains to prove that the~$\tilde{s}$-$\tilde{t}$-paths of cost zero are exactly the
cheapest $\tilde{s}$-$\tilde{t}$-path in~$\tilde{D}$. Together with the result of the previous
paragraph this then implies that an $s$-$t$-path~$P$ is admissible if and only if~$\tilde{P}$ is
a cheapest~$\tilde{s}$-$\tilde{t}$-path, which, by definition of~$A'$, is equivalent to the statement of the lemma.

We thus need to show that an arbitrary~$\tilde{s}$-$\tilde{t}$-path~$\tilde{P}$ in~$\tilde{D}$
has non-negative cost. The first and last arcs of~$\tilde{P}$ are~$\tilde{s}s$, with~$s\in\sources$,
and~$t\tilde{t}$, with~$t\in\sinks$, respectively. Inbetween lies an~$s$-$t$-path~$P$ whose cost
is at least the cost of a cheapest~$s$-$t$-path~$\bar{c}_{st}$. Thus, 
\[
c(\tilde{P})=c_{\tilde{s}s}+c(P)+c_{t\tilde{t}}\geq -y^*_s+\bar{c}_{st}+y^*_t\geq0 
\]
where the
last inequality immediately follows by dual feasibility~\eqref{eq:dual}. This concludes the proof.
%
%
\end{proof}
Summarizing, we can state the following algorithm which reduces the solution of the quickest
minimum cost transshipment problem to computing a quickest transshipment on a suitably chosen
subnetwork.
\pagebreak[2]
\begin{algorithm}\label{algo}
Quickest Minimum Cost Transshipment
\begin{enumerate}
\item for all~$st\in\sources\times\sinks$, compute cost~$\bar{c}_{st}$ of cheapest $s$-$t$-path in~$D$
\item determine optimal dual solution~$y^*$ to transportation problem on~$\bar{D}$
\item\label{step:cheapest-path-network} find arc set~$A'$ via cheapest-paths subnetwork of extended network~$\tilde{D}$
\item compute quickest transshipment on~$D'=(V,A')$
\end{enumerate}	
\end{algorithm}
With respect to the cheapest-paths subnetwork of~$\tilde{D}$ in
Step~\ref{step:cheapest-path-network}, we refer to our short discussion of how to compute a cheapest-paths subnetwork at the end of Section~\ref{sec:quickest-min-cost-flow}.
\begin{theorem}\label{thm:main}
Algorithm~\ref{algo} correctly solves the quickest minimum cost transshipment problem. Its
running time is dominated by the quickest transshipment computation in its final step.
\end{theorem}
\begin{proof}
By Lemma~\ref{lem:active} and Observation~\ref{obs:active}, any transshipment over time in the
subnetwork~$D'=(V,A')$ yields a minimum cost transshipment over time in~$D$, and any minimum cost
transshipment over time in~$D$ must only use arcs in~$A'$ and thus lives in~$D'$. In particular, a
quickest transshipment in~$D'$ yields a quickest minimum cost transshipment in~$D$.
\end{proof}
The currently fastest known quickest transshipment algorithm is due to \citet{SchloterST22} and
achieves a running time of~$\tilde{O}(m^2k^5+m^3k^3+m^3n)$,
where~$k\coloneqq|\sources\cup\sinks|$ and the~$\tilde{O}$-notation omits all poly-logarithmic
terms.
\section{Concluding Remark}
Flows over time with costs can be seen as a bicriteria problem, with time horizon and total cost as
two conflicting objectives.
%
%
The efficiently solvable quickest minimum cost flow problem asks for a solution at one end of the
corresponding tradeoff curve, while the NP-hard minimum cost quickest flow problem targets the other
end. As \citet{KlinzW95,KlinzW04} point out, computing minimum cost flows over time with bounded
time horizon is NP-hard.
This essentially means that solutions corresponding to intermediate points on the tradeoff curve
are NP-hard to find. In view of this, the fact that the minimum cost quickest
flow/transshipment problem can be solved efficiently thus reveals a kind of singularity in the
complexity landscape of minimum cost flows over time.

\subsection*{Acknowledgement} This note is dedicated to the memory of my friend and colleague Gerhard Woeginger, in thankful admiration for his outstanding scientific contributions, his fine taste of problems, his unfailingly inspiring lectures, and, last but not least, for his great sense of humor. 


\end{document}